\documentclass{article}

\usepackage{times}
\usepackage{graphicx} 
\usepackage{fullpage}
\usepackage{url}

\usepackage{subfigure} 
\usepackage{natbib}
\usepackage{algorithm}
\usepackage{algorithmic}

\usepackage{caption}
\usepackage{latexsym,amsmath,amssymb}
\usepackage{url}
\usepackage{bm}
\usepackage{xcolor}
\usepackage{algorithm}
\usepackage{algorithmic}
\usepackage{amsmath}
\usepackage{amssymb}
\usepackage{amsthm}
\usepackage{graphicx}
\usepackage{hyperref}
\hypersetup{
bookmarksnumbered
}

\def\R{\mathbb{R}}
\newcommand{\nnz}{\textrm{nnz}}
\newcommand{\trans}[2][]{#2^{#1\!\top}}

\newcommand{\E}{{\mathcal E}}

\newcommand{\tr}{\mathrm{Tr}}
\newcommand{\expect}{{\bf E}}
\newcommand{\prob}{{\bf Pr}}

\newcommand{\T}{{\mathcal{T}}}

\newtheorem{theorem}{Theorem}
\newtheorem{lemma}[theorem]{Lemma}
\newtheorem{claim}[theorem]{Claim}
\newtheorem{definition}[theorem]{Definition}

\newtheorem{remark}[theorem]{Remark}
\newenvironment{proofof}[1]{\noindent{\bf Proof of #1:}}{$\qed$\par}

\begin{document}

\title{How to Fake Multiply by a Gaussian Matrix}

\author{Michael Kapralov\thanks{michael.kapralov@epfl.ch}\\EPFL \and Vamsi K. Potluru\thanks{vamsi\_potluru@cable.comcast.com}\\Comcast Cable  \and David P. Woodruff\thanks{dpwoodru@us.ibm.com}\\IBM Research}


\maketitle

\begin{abstract}
Have you ever wanted to multiply an $n \times d$ matrix $X$, with $n \gg d$, 
on the left by an $m \times n$ matrix $\tilde G$ of i.i.d. Gaussian random variables, 
but could not afford to do it because it was too slow? In this work 
we propose a new randomized $m \times n$ matrix $T$, for which one can
compute $T \cdot X$ in only $O(\nnz(X)) + \tilde O(m^2 \cdot d^{3})$ time, for which
the total variation distance between the distributions $T \cdot X$ and $\tilde G \cdot X$
is as small as desired, i.e., less than any positive constant. Here
$\nnz(X)$ denotes the number of non-zero entries of $X$. Assuming 
$\nnz(X) \gg m^2 \cdot d^{3}$,
this is a significant savings over the na\"ive $O(\nnz(X) m)$ time to compute $\tilde G \cdot X$.
Moreover, since the total variation distance is small, we can provably use $T \cdot X$
in place of $\tilde G \cdot X$ in any application and have the same guarantees as if we were using
$\tilde G \cdot X$, up to a small
positive constant in error probability. 
We apply this transform to nonnegative matrix factorization (NMF) and support vector machines (SVM).

\end{abstract}

\section{Introduction}
\label{sec:intro}
One approach to handle high dimensional data, often in the form of a matrix, 
is to first project the data to a much lower dimensional subspace. 
This is an example of sketching and the last decade has seen a systematic study of this approach.
A linear sketch of a matrix replaces the original matrix by a smaller matrix
which is often obtained by a random projection of the original 
matrix~ (see, e.g., \cite{woodruff2014sketching} for a survey).  
Random projections have been successfully applied to speed up least squares regression and have been implemented with
remarkable success~\cite{avron2010blendenpik}. 
This is impressive considering the fact that these solvers have been highly optimized 
over the last few decades, exploiting both algorithmic improvements and machine dependent optimizations.

Many of these works rely on {\it fast projection matrices}, such as the Subsampled Randomized Hadamard Transform
or the CountSketch, the latter being particularly well-suited for sparse data (see, e.g., \cite{woodruff2014sketching}
and references therein). However, there are certain applications for which multiplying by a Gaussian matrix is 
the only way that is known to reduce the dimensionality of the data. This arises mainly because the application
requires rotational symmetry, which is often not preserved by other fast transforms, or because additional properties,
such as spreading out a sparse vector to a vector with non-spiky elements, do not hold for transforms like
CountSketch (some of these hold for the Fast Hadamard Transform, but the latter are not known to be able to exploit sparsity). We give
two such applications below, one to nonnegative matrix factorization (NMF), and one to support vector machines (SVM). 

\subsection{Our Results}
{\bf A New Randomized Transform.} 
In this work we propose a new randomized transform $T$, which we call the CountGauss. It is simply a product of a
CountSketch matrix and a Gaussian matrix.
That is, given an $n \times d$ matrix $X$ which we would like to 
multiply by an $m \times n$ matrix $\tilde G$ of Gaussians, we instead let $T = G \cdot S$, where $S$ is a
$B \times d$ CountSketch matrix  
where $B = \tilde{O}(d^{2}m)$, and $\tilde G$ is an $m \times B$ matrix of i.i.d. 
Gaussians.
Recall that a CountSketch matrix $S$ satisfies that each column of $S$ has only a single non-zero
entry chosen in a uniformly random position. That non-zero is $1$ with probability $1/2$, and $-1$ with
probability $1/2$. The columns of $S$ are independent of each other. 
Importantly, computing $S \cdot X$ can be done in $O(\textrm{nnz}(X))$ time, and this significantly reduces
the number of rows of $X$. Then computing $G \cdot (S \cdot X)$ can now be done in $\tilde{O}(m^2 d^{3})$ time.
While such a composition of matrices has been used before in the context of subspace embeddings for regression,
see, e.g., \cite{cw13}, here we show a new property of this composition - the distribution of $G \cdot S \cdot X$
looks like the distribution of $\tilde G \cdot X$! Formally, the statistical distance between the two distributions
is smaller than any positive constant.

{\it Therefore, in any application which uses $\tilde G \cdot X$, if we replace
$\tilde G \cdot X$ with $G \cdot S \cdot X$, 
then if $p$ is the success probability of the old algorithm, then the success probability
of the new algorithm is at least $p - \delta$, where $\delta > 0$ is an arbitrarily small constant.}

We now give applications. 
\\\\
{\bf Non-negative Matrix Factorization.}
Learning low rank structures and representations is a fundamental problem in machine learning. 
With the rise of data-driven decision making, many businesses, government agencies, 
and scientific laboratories are collecting increasingly 
large amounts of data each day. 
For instance, the large Hadron Collider (LHC) 
experiments represent about $150$ million sensors acquiring around $40$ million samples
per second. Even working with $0.001$ percent of the sensor data, the data flow from all four 
LHC experiments is around $25$ petabytes per day~\cite{brumfiel2011down}. 
This means the traditional approach of storing the data, and then processing it 
later, may be infeasible.
One approach would be to subsample the incoming streams. However, we may lose valuable 
information in the form of infrequent events.

We use our transform to solve the nonnegative matrix factorization (NMF) problem. 
Previous approaches~\cite{damle2014random,benson2014scalable,tepper2015compressed} 
have used random matrices for the projection. However, these approaches can be slow if the
dimensionality of the data is high since they rely on multiplying by Gaussian matrices, 
e.g., for natural images or structural Magnetic Resonance 
Imaging brain scans. Recent work by Smola et al.~\cite{le2013fastfood} have shown that
sometimes dense random Gaussian matrices can be replaced by faster transforms, 
and moreover, each row of the transform is equally likely to be in any 
direction on the unit sphere. To show the correctness of the NMF algorithm, however, 
we need a much stronger property than this, namely that any small
subset of rows of the transform has the property that its product with
a fixed matrix $X$ has low variation distance to the distribution of a product
of a Gaussian matrix with $X$. These latter properties, of having a fast transform with 
{\it equal representation of directions on the sphere}, 
do not seem to have been exploited in the context of NMF. Our transformation, since it
has low variation distance to multiplying by a Gaussian matrix, directly applies here and
we can use existing analysis. 

We note that the classical way of speeding up Gaussian
transforms via the Fast Hadamard or Fast Fourier Transform (see, e.g., \cite{t11}) 
do not work in this context, since they miss large sections of the sphere, and we provide 
a formal counterexample in Section~\ref{sec:counterexample}. Intuitively, while it is fine to miss directions along large 
sections of the sphere to approximate the norm of a vector,
it is not fine to miss directions for NMF, where the corresponding 
polytope partitions the sphere into a small number of caps, and each cap should have a random 
direction chosen from it. 
\\\\
{\bf Support Vector Machines.}
We also apply random projections to the support vector machines (SVM) problem. 
Previously, the CountSketch (CW)~\cite{clarkson2013count} projection and random Gaussian (RG) projection  
have been applied to the linear SVM problem.
Despite Countsketch being much faster than the 
Gaussian projection, the overall running time of projection together with the SVM solver was similar 
for both 
projections~\cite{paul2014random}, since the training of the projected data was faster when using
Gaussian projections. Our projection combines the CW matrix with 
a smaller Gaussian matrix thereby getting the best of both worlds --- similar projection 
time as CountSketch and similar 
Gaussian properties of RG that are useful for SVM.
\\\\
{\bf Experiments.}
We empirically validate our results for both NMF and SVM applications. 
For NMF, we give an experimental evaluation by comparing with state-of-the-art algorithms 
such as SPA~\cite{gillis2014fast}, XRAY~\cite{kumar2013fast}, na\"ive random projections~\cite{damle2014random} ,
structured Gaussian random projections~\cite{tepper2015compressed}, and Tall-Skinny QR 
factorization~\cite{benson2014scalable} for NMF problems with applications to breast cancer, 
flow cytometry, and climate data.
Also, we show experimental speedups using our projection when combined with linear SVM solvers
 for document classification problems~\cite{paul2014random}. 

\section{A New Randomized Transform}\label{sec:newTransform}
A CountSketch matrix $S\in \mathbb{R}^{B\times n}$ is a matrix all of whose rows have exactly one nonzero in a uniformly random location, and the value of the nonzero element is independently chosen to be $-1$ or $+1$ with equal probability. We denote the number of rows in the CountSketch matrix by $B$.

We prove the next theorem\footnote{In the conference version of the paper the authors made the stronger claim that $B\approx \sqrt{m} d^2$ suffices for the same guarantee, but that was in error. We provide a lower bound (Lemma~\ref{lm:lb}) showing that the result of Theorem~\ref{thm:main} is essentially tight.}, which gives the formal guarantees of our new transform.
\begin{theorem}\label{thm:main}
There exists an absolute constant $C>0$  such that for every $\delta\in (0, 1)$, every integer $m\geq 1$ and every matrix $U\in \R^{n\times d}$ with orthonormal columns if $B\geq \frac1{\delta^2}C d^2\cdot m$,  $S\in \mathbb{R}^{B\times n}$ is a random CountSketch matrix, and $G\in \R^{m\times B}$ and $\tilde G\in \R^{m\times n}$ are matrices of i.i.d. unit variance Gaussians, then the total variation distance between the joint distribution $G S U$ and $\tilde GU$ is less than $\delta$.
\end{theorem}
The proof is given in Section~\ref{sec:countsketch}.
We note that Theorem~\ref{thm:main} applies to matrices $U$ with orthonormal columns. This is sufficient for applying our transform to an arbitrary matrix $X$, since we can write $X=UR$, where the columns of $U$ form an orthonormal basis for the range of $X$, and apply the theorem to $U$.
Since $G\cdot S\cdot U$ is close to $\tilde G\cdot U$ in total variation distance, $G\cdot S\cdot UR=G\cdot S\cdot X$ is close to $\tilde G\cdot UR=\tilde G\cdot X$ in total variation distance as well.
We note that the role of $d$ and $n$ in Theorem~\ref{thm:main} is swapped in comparison to our notation for the application to NMF below. The notation in Theorem~\ref{thm:main} is more consistent with the numerical linear algebra literature, and we thus prefer to state the theorem in this form.

We now present the intuition behind Theorem \ref{thm:main}. Consider the distribution of the first row of the two matrices, namely $\tilde G U$ versus $G SU$. Both random variables are Gaussians in dimension $d$, but while the former is an ideal isotropic Gaussian, the latter, despite being Gaussian, has correlated entries. The correlations between the entries are due to the fact that the CountSketch matrix $S$ is not a perfect isometry: the correlation is given exactly by $U^TS^TSU$, which is the identity in expectation, but not for most realizations of $S$. In order to show that these two distributions are close in total variation distance, it would suffice to argue that the covariance matrix $U^TS^TSU$ is sufficiently close to the identity. This is exactly how the proof of Theorem \ref{thm:main} proceeds, which fixes an $S$ for which $U^TS^TSU$ is sufficiently close to the identity, using a so-called ``approximate matrix product'' theorem in the linear algebra community. After fixing such an $S$,
one can use that the rows of $G \cdot S$ and the rows of $\tilde{G}$ are independent, and then bound the variation distance between individual rows of $G \cdot S$ and of $\tilde{G}$. For the latter, it is convenient to work with Kullback-Leibler divergence (KL divergence) which is additive over product spaces; here we bound the KL divergence between a standard multivariate Gaussian and one with covariance matrix $U^TS^TSU$. 

The result of the theorem is essentially tight -- we show in Lemma~\ref{lm:lb} that if $B\leq md^2/(C\log d)$ for a sufficiently large constant $C>0$, then the total variation distance between $GU$ and $\tilde G S U$ is lower bounded by $e^{-1}(1-o(1))$. The proof looks at a particular matrix $U$ with large leverage scores, and builds a distinguisher.

\section{Preliminaries for the Applications}\label{sec:prelim}
A few applications of our new randomized transform are NMF and SVM, which we now formally define. 
\subsection{Nonnegative Matrix Factorization}
Given a nonnegative matrix $X$ of size $d \times n$, we would like to approximate it as a product of nonnegative matrices as follows:
$X \approx WH,$ 
where $W$ is of size $d \times k$ and $H$ is $k \times n$. 
This problem was studied by Paatero and Tapper~\cite{paatero1994positive} under the name of positive matrix factorization and gained a wider popularity through the work of Lee and Seung~\cite{lee2001algorithms}.
NMF arises in a wide range of problems and application domains
such as curve resolution in chemometrics and document clustering; further references can be found in~\cite{arora2012computing}. 
Various extensions to the original model to incorporate domain knowledge such as sparsity, orthogonality~\cite{ding2006orthogonal}, and under-approximation~\cite{gillis2010using} have also been studied.   
Commonly used measures of approximation include the Frobenius norm, Itakuro-Saito (IS), and Bregman divergence with applications in image processing, speech and music analysis~\cite{yilmaz2011generalised} among other places.
Typical algorithms use alternating minimization to solve the non-convex objective function arising from NMF.

Until recently, the complexity of the NMF problem was unknown. Vavasis established that the NMF problem is NP-hard~\cite{vavasis2009complexity}. However, if the data satisfies the separability condition, a condition 
introduced by Donoho and Stodden~\cite{donoho2003separable}, 
then tractable algorithms exist and have been recently proposed by Arora et al.~\cite{arora2012computing,recht2012factoring}.
Formally, a nonnegative matrix $X$ is $k$-separable if it satisfies the following condition:
$X = X_{I} H$, 
where $I$ is an index set of size $k$ 
corresponding to the columns of the data matrix $X$.
Geometrically, this assumption implies that the columns of $X$ 
lie in a cone generated by the $k$ selected
columns of $X$ indexed by $I$. 
One can view these $k$ selected columns as the extreme points of a polytope containing all other columns.
In practice, $k$ is much smaller than both $d$ and $n$. We will assume $k$-separability.

Given $X_{I}$, one can solve for $X$ by solving a nonnegative
least squares problem \cite{damle2014random}, 
and therefore our focus is on finding $X_{I}$, or equivalently, the index
set $I$ of extreme points of the point cloud formed by the columns of $X$. 

To understand the guarantees of our algorithm, we define a few geometric notions also
used in \cite{damle2014random}, which we refer to for more background. 
The {\it normal cone}
of a convex set $C$ at a point $x$ is the cone
$$N_C(x) = \{w \in \mathbb{R}^d \mid w^T(y-x) \leq 0 \textrm{ for
any } y \in C\},$$
that is, it is the cone defined by the outward normals of supporting hyperplanes at the point
$x$. One can define a measure $\omega(K)$ on any cone $K$, which for full-dimensional
cones $K$ satisfies $\omega(K) = \Pr[\theta \in K \cap S^{d-1}]$ where $\theta$ is a uniformly
random point on the sphere $S^{d-1}$ in $d$ dimensions. This measure is known as the 
{\it solid angle} of $K$. For any convex polytope $C$, 
if $P$ is the set of its extreme points,
then $\sum_{p \in P} \omega(N_C(p)) = 1,$
that is, the solid angles of the normal cones at the extreme points sum to $1$.
If we label the points $p_i \in P$, we will use the shorthand
$\omega_i = \omega(N_C(p_i))$. 

A key property we will use is that for a unit vector $u$ and a convex set $C$, 
the maximum inner product of $u$ with any point $p \in C$ is  
achieved by an extreme point $p$ of $C$. Moreover, the maximum
is achieved by the extreme point $p$ precisely when $u \in N_C(p)$. 
This follows since the inner product with a fixed vector $u$ is a linear function,
which is maximized by an extreme point for any convex set. These conditions also
hold if we replace maximum with minimum. 

Our results, as in \cite{damle2014random}, depend on the condition number
$\kappa = \frac{1}{k \log \left (\frac{1}{\max_i 1- 2\omega_i} \right )}.$
The larger $\kappa$ is, the more pointed the polytope defined by the columns
of $X$ is, whereas if $\kappa$ is small, the polytope has ``fatter'' vertices. 

\subsection{Support Vector Machines}
Given a dataset of samples and labels $\{x_i,y_i\}_{i=1}^N$ where $x_i$ corresponds to sample $i$ and $y_i$ the
corresponding label belonging to one of two classes denoted by $\{-1, 1\}$, we would like to find a maximum-margin hyperplane that separates
the two classes. 
The primal form for the linear SVM problem is as follows: 
\begin{align}
\min_{w} \frac{1}{2} \|w\|_2^2 + \frac{C}{N} \sum_{i=1}^N \max(0, 1 - y_i \langle w, x_i \rangle)
\end{align}
where $C$ is the soft margin parameter which allows for mis-classfication errors in the dataset and $w$ is
the maximum margin hyperplane that we are learning from the data.
The dual form for the linear SVM problem is given as follows:
\begin{align}
\max_{0 \le \alpha \le C} 1^T \alpha - \frac{1}{2} \alpha^T Y X X^T Y \alpha 
\end{align}

Previously~\cite{paul2014random} have shown that the margin (hyperplane) and minimum enclosing ball of the
original data are preserved after projection up to a multiplicative factor. However, in their original formulation
it is possible to just replace all points with zero to achieve the same guarantee. 
We strengthen the theorems by requiring that the projected data
upper bound the objective of the original data. The details are given in Section~\ref{sec:svm-theorem}.

\section{Application to NMF}
We consider the separable NMF problem as defined in Section \ref{sec:prelim}. 
We first review an algorithm proposed by ~\cite{damle2014random}. 
%
Their algorithm involves the computation of $\tilde GX$ where $\tilde G$ 
has dimensions $m \times d$ for a parameter $m$, and the 
entries are distributed independently as $N(0,1)$ random variables.
Notice that we need to first compute the $m \times d$ random matrix $\tilde G$
which
is itself dense. We also need to compute the matrix product $\tilde GX$ 
with the input data. This is computationally expensive and is of order
$O(mnd)$ in practice. Fast matrix multiplication routines
~\cite{cw90,w12} can be used in theory, but the time will still be 
at least $k^{\omega-2}nd$, where $\omega \approx 2.376$ is the exponent of fast matrix
multiplication.
Instead, we propose to use our new transform 
to 
significantly speed up the computations for extracting the
extreme points in the dataset.
Note that both approaches are easily amenable to distributed-data settings by simply sharing
the seed of the random number generator which allows identical matrix transformations on all the computational nodes.
Our new algorithm is called \textbf{C}ount \textbf{G}auss NMF or \textbf{CountGauss} and is as follows:
\begin{algorithm}
    \caption{$\textrm{\textbf{C}ount\textbf{G}auss NMF (\textbf{CG})}$}
\label{alg:sketchNMF}
Initialize the index sets $I_{max},I_{min}$ to empty.\\
\begin{enumerate}
\item Let $T=G \cdot S$ where $G$ is an $m \times B$ 
matrix of i.i.d Gaussians, and $S$ is a $B \times d$ CountSketch
matrix. Here $B = Cn^{2}m/ \delta^2$.  
\item Compute the product $Z = TX$.
\item Find the indices which give the maximum and 
minimum across each row of $Z$ corresponding to $I_{max},I_{min}$
\end{enumerate}
\end{algorithm}
Instead of using Gaussian random matrices for the projection, we approximate them by the following projection
matrix $T=G \cdot S,$ where the matrices are defined in 
Algorithm \ref{alg:sketchNMF}. 

Consider the convex polytope defined by the columns of $X$ and their
negations. As defined in Section \ref{sec:prelim}, we assume 
$k$-separability, namely, that there are $k$ columns of $X$, indexed by $I$, 
for which $X = X_I H$ for a nonnegative matrix $H$. The columns 
of $X_I$ are the extreme 
points of a convex polytope $C$. By definition of an extreme point of a 
convex polytope, the indices found in step 3 of
Algorithm \ref{alg:sketchNMF} belong to the index set $I$. 

Damle and Sun show the following.
\begin{theorem}(Theorem 3.3 of \cite{damle2014random})\label{thm:damleMain}
Consider a modification to Algorithm \ref{alg:sketchNMF} in which we
replace $T$ by an $m \times d$ matrix of i.i.d. $N(0,1)$ random variables,
where $m = \kappa k \log(\frac{k}{\delta})$, 
where recall $\kappa = \frac{1}{k \max_i 1-2\omega_i}$ is the condition number. 
Then the probability that
the output $I_{min} \cup I_{max}$ of Algorithm \ref{alg:sketchNMF} contains
the index set $I$ of extreme points of $X$ is at least $1-\delta$.
\end{theorem}

Using Theorem \ref{thm:main}, we analyze the performance of 
Algorithm \ref{alg:sketchNMF}.

\begin{theorem}\label{thm:algAnalysis}
Let $\delta > 0$ be given. 
Suppose in Algorithm \ref{alg:sketchNMF} we set 
the parameter $m = \kappa k \log(\frac{k}{\delta})$, where 
$\kappa = \frac{1}{k \max_i 1-2\omega_i}$ is the condition number, and choose $B\geq C n^2 m / \delta^2$ for a sufficiently large constant $C>1$ as per Theorem~\ref{thm:main}. 
Then the probability that
the output $I_{min} \cup I_{max}$ of Algorithm \ref{alg:sketchNMF} contains
the index set $I$ of extreme points of $X$ is at least $1-2\delta$.
\end{theorem}
\begin{proof}
Let $I$ be the index set of extreme points of the polytope defined by
the columns of $X$. By definition of an extreme point, 
in each iteration of step 4 of the algorithm, we 
add an index $i \in I$ to $I_{max}$ and an index $j \in I$ to $I_{min}$ (since we are taking
the inner product with a linear function). 
Therefore, the behavior of Algorithm \ref{alg:sketchNMF} is the same
if we instead, in each invocation of step 2, compute the product
$Z = T X_I$. 

By our assumption on $m$, since $X_I$ is a $d \times k$ matrix we may
apply Theorem \ref{thm:main}, with the role of $n$ and $d$ in that theorem swapped,
to obtain that the variation
distance of the distributions of $T X_I$ and $\tilde G X_I$ is at most
$1 - \delta$, where $\tilde G$ is a matrix of i.i.d. $N(0,1)$ random
variables. Therefore, we can apply Theorem \ref{thm:damleMain} to conclude by a union bound that
the output of Algorithm \ref{alg:sketchNMF} contains the set $I$ with probability
at least $1-2\delta$.
\end{proof}
We obtain the same guarantee as in Theorem \ref{thm:damleMain} with considerably faster computation time.
Indeed, our
matrix product $TX$ can be computed in
$O(\textrm{nnz}(X)) + \tilde{O}(m^{2}n^{3})$ time using our transform $T$, 
as opposed to the $O(dnm)$ time needed in
\cite{damle2014random} to compute the product $ \tilde{G} X$ for a matrix
 $\tilde{G}$ of i.i.d. Gaussians. This is significant when $d$ is very large. 
%

{\bf Distributed Environments:} Our results naturally provide 
solutions to NMF in a distributed environment
in which the columns of $X$ are partitioned across multiple servers. 
Indeed, the servers can agree upon a short random seed of length $O(d)$ words
to generate $T$. Each server can then compute its local sets
$I_{max},I_{min}$, and send them to a coordinator who can find the global
maxima and minima. 
\begin{figure*}[ht!]
\includegraphics[width=0.47\linewidth]{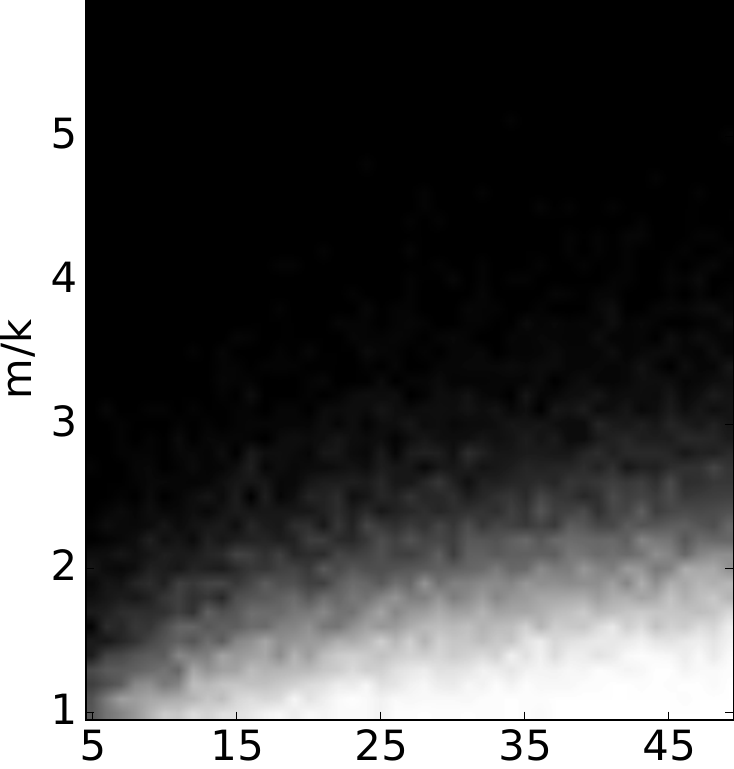}
\hfill
\includegraphics[width=0.42\linewidth]{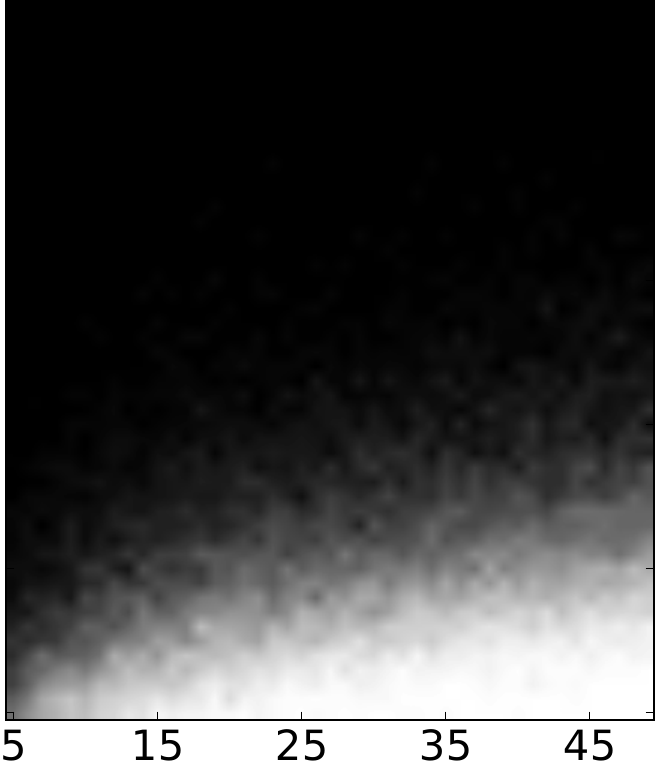}
\caption{The fraction of trials in which the CG algorithm correctly extracted all ``k' extreme
    points. For each value of $k$ and $m$, we generate $500$ matrices, such that the data matrix is of size $1000 \times 500$, and
show how often we successfully recovered the original anchors (black indicates success). (Left) We contrast Gaussian random projections (GP) with (right) our algorithm countGauss. Note that we recover the anchors with a similar success rate as GP.}
\label{fig:exact_recover}
\end{figure*}
\begin{figure*}[ht!]
\begin{center}
\includegraphics[width=0.97\linewidth]{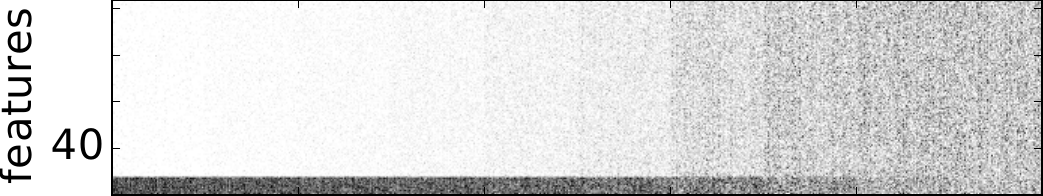}
\includegraphics[width=0.98\linewidth]{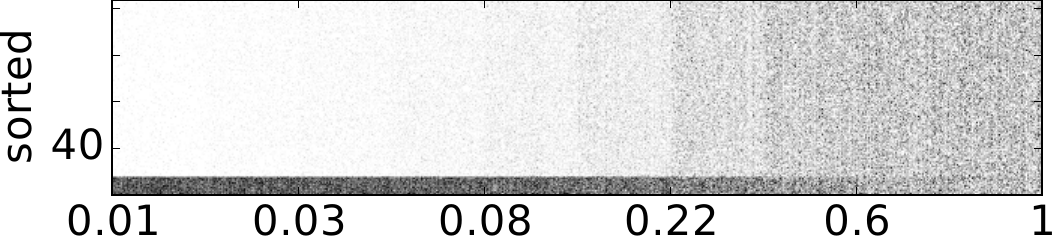}
\end{center}
\caption{We show the scree plots at $20$ noise levels
and notice that there are sharp transitions at $20$ corresponding to the rank of the data. (Top) Gaussian random projections and (bottom) our algorithm countGauss are applied to the dataset. For each
noise value in $\{0.01, 0.02, 0.03, 0.05, 0.08, 0.12, 0.22, 0.36, 0.6,1\}$, we generate $100$ datasets. At higher noise levels, we note that
both the algorithms GP and CG have most of the features active and there is no longer a sharp transition at $20$.}
\label{fig:noise_recover}
\end{figure*}

\section{Other Related work}
Over the last couple of years, many approaches have been proposed to solve the separable-NMF problem. \\
\textbf{XRAY}
Selects the anchors one at a time by expanding a cone until all columns in the dataset are contained in it. At each step, XRAY finds the datapoint (column) which maximizes the inner product with the current residual matrix. It then computes the
residual matrix corresponding to the new set of anchor points~\cite{kumar2013fast}.\\
\textbf{SPA}
Successive projection algorithm~\cite{araujo2001successive,gillis2014fast} is a family of recursive algorithms where the projections are given by strongly convex functions. \\
\textbf{TSQR}
Use tall and thin QR factorization when the number of rows/features is large~\cite{benson2014scalable}. 
This approach is especially attractive when the number of features is really large ($\gg 10^6$) and the number of samples
is small ($<10^5$).\\
\textbf{SC}
In~\cite{tepper2015compressed}, an algorithm similar to the one proposed by Damle and Sun
~\cite{damle2014random} is proposed. 
The 
difference is that instead of choosing a Gaussian or FastFood projection matrix, the projection is chosen
to be a matrix which depends on $X$ (data dependent projection), namely, one that is found via the subspace power
iteration (see Figure 3 of \cite{tepper2015compressed}). This approach is expensive in the case of distributed settings
since the projection matrix depends on all the samples.
\begin{figure}[ht!]
\begin{center}
    \includegraphics[width=0.97\linewidth]{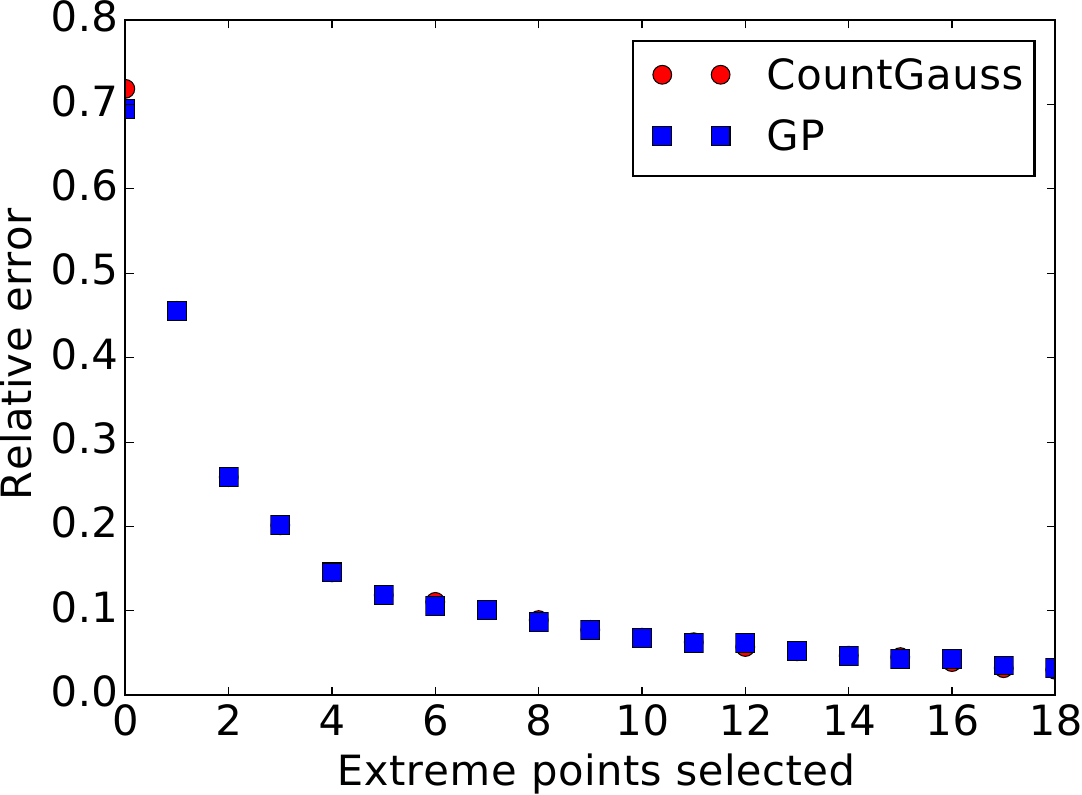}
\end{center}    
\caption{Relative reconstruction error as a function of the anchors selected by the two algorithms CountGauss and GP is shown.
They are remarkably similar.}
\label{fig:breast_comp}
\end{figure}
\begin{figure}[ht!]
\includegraphics[width=0.47\linewidth]{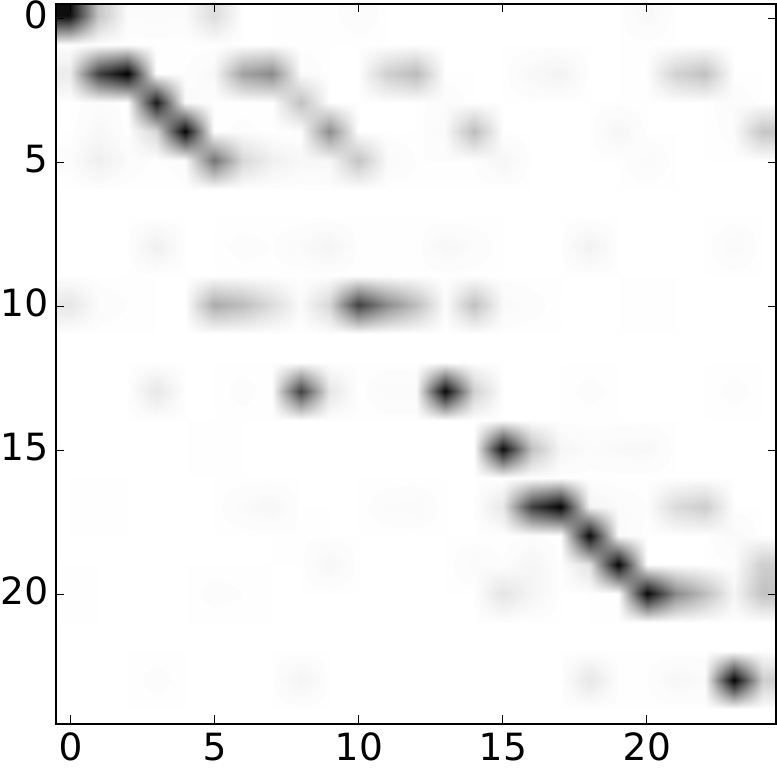}
\hfill
\includegraphics[width=0.47\linewidth]{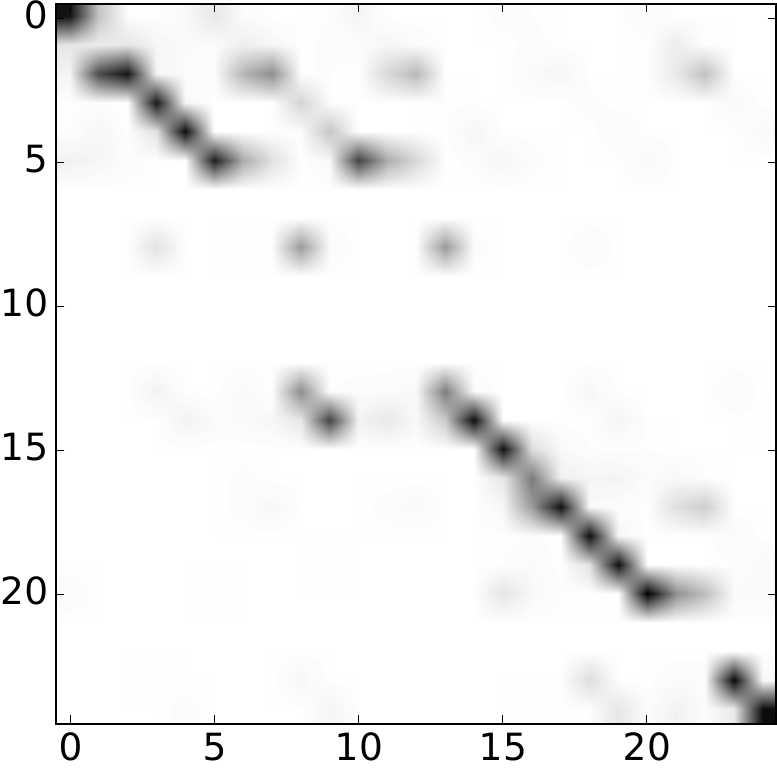}
\caption{Coefficient matrices $H$ are shown for the two algorithms GP and CountGauss for the flow cytometry data when $k$ is
set to $16$. The coefficients tend to be clustered near the diagonal as has been previously observed.} 
\label{fig:flow}
\end{figure}
\begin{figure*}[ht!]
%
    \begin{minipage}[b]{0.6\linewidth}   
    \includegraphics[width=1\linewidth]{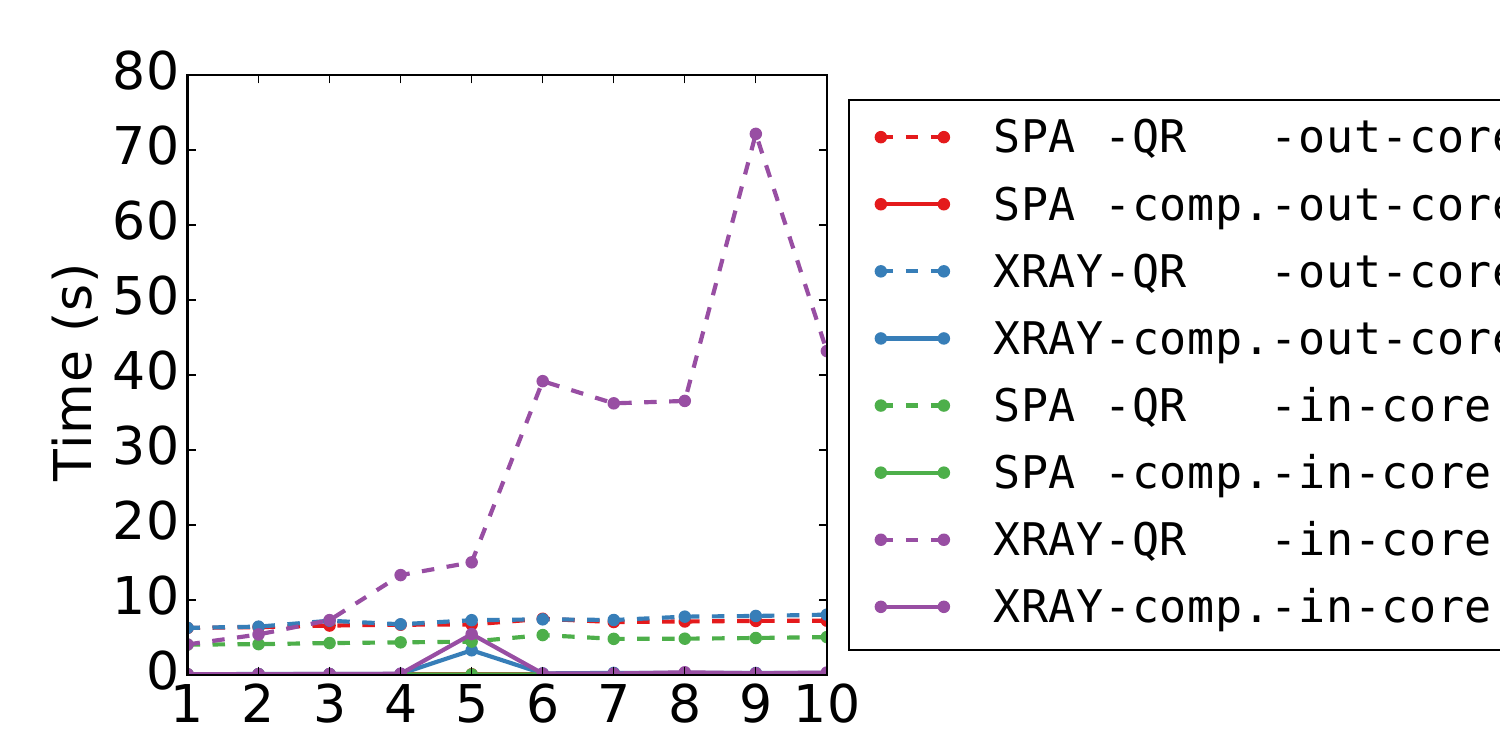}
 (A) \qquad      Running times using CountGauss (comp==CG)
\end{minipage}
\begin{minipage}[b]{0.6\linewidth}   
\includegraphics[width=1\linewidth]{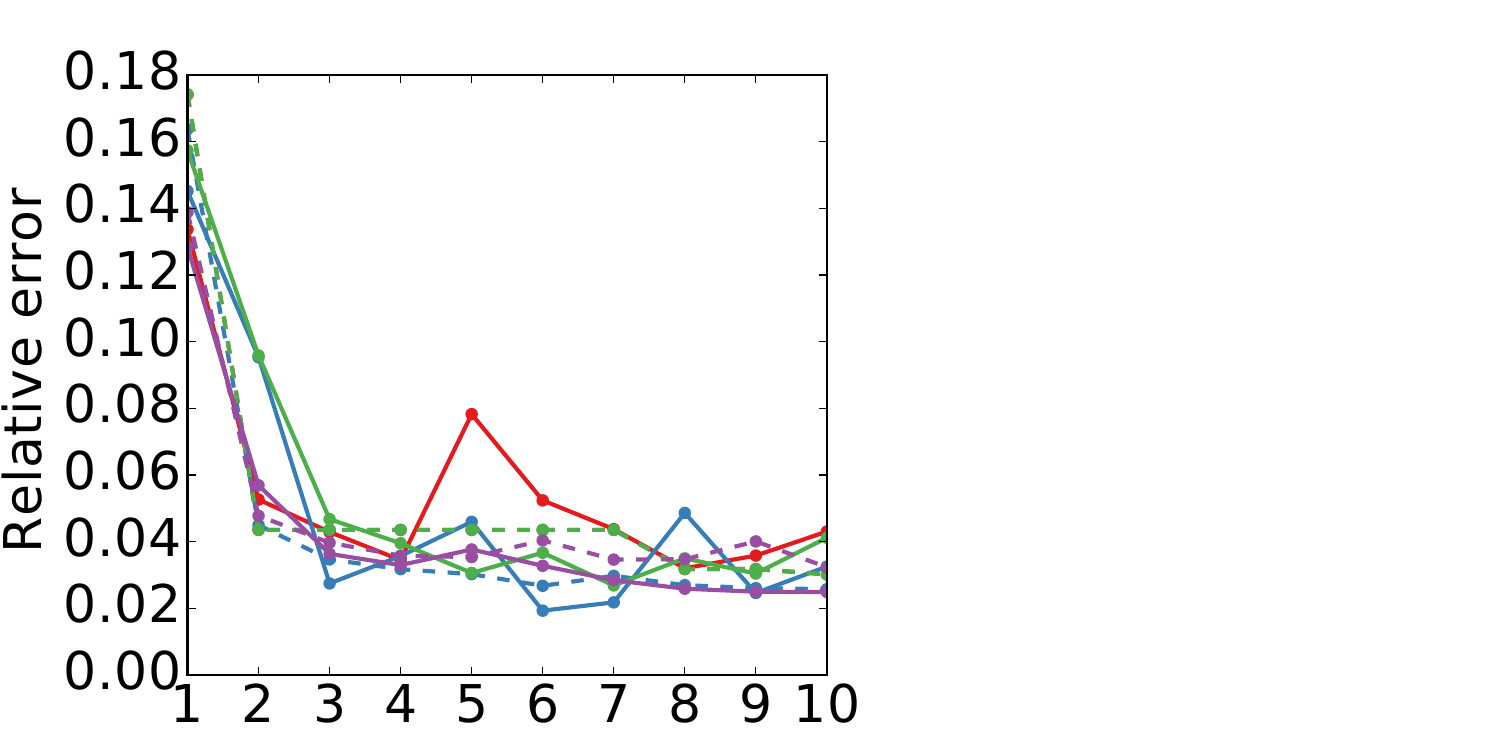}\\
(B)  \qquad     Relative error using CG
\end{minipage}
\caption{Extracting columns with CG versus using QR factorization. Note that the QR-based methods
are optimized for tall-and-skinny matrices and tend to do poorly for fat matrices. Note that CG (and SC) tends to perform
well since it is based on random projections and is at least an order of magnitude faster than QR-based methods.}
\label{fig:climate}
\end{figure*}

\begin{table*}[ht!]
\begin{center}
\begin{tabular}{|l| c|| c | c || c | c|| c | c || c|| r|| }
\hline
      \multicolumn{2}{|c||}{Test}   &     \multicolumn{2}{c||}{Proj}       &     \multicolumn{2}{c||}{SVMf}    &   \multicolumn{2}{c||}{Margin}   &  Proj &  Algo \\
\hline
  mean  &  std &    mean &   std  &  mean &    std  &   mean &    std  &  & \\
\hline                   
 17.92 & 11.29  &  0.0000 &  0.0000 &  0.89 &  0.38  &  2.1057 &  3.9391 &  full & full \\
 24.71 & 12.60  &  0.0086 &  0.0042 &  0.38 &  0.19  &  1.6792 &  3.5714 &   128 & countSketch \\
 25.27 & 13.08  &  0.0216 &  0.0047 &  \textbf{0.16} &  \textbf{0.11}  &  1.6277 &  3.5634 &   128 & countGauss  \\
 25.07 & 13.20  &  0.3676 &  0.1569 &  0.49 &  0.20  &  1.7143 &  3.7143 &   128 &         RG  \\
\hline                    
 17.92 & 11.29  &  0.0000 &  0.0000 &  0.89 &  0.38  &  2.1057 &  3.9391 &  full & full        \\
 22.56 & 12.42  &  0.0082 &  0.0036 &  0.54 &  0.21  &  1.8778 &  3.6709 &   256 & countSketch \\
 24.34 & 12.23  &  0.0565 &  0.0091 &  \textbf{0.21} &  \textbf{0.07}  &  1.8722 &  3.7389 &   256 & countGauss \\
 23.66 & 12.86  &  0.8178 &  0.3286 &  0.98 &  0.35  &  1.8895 &  3.6747 &   256 &         RG \\
\hline                   
 17.92 & 11.29  &  0.0000 &  0.0000 &  0.89 &  0.38  &  2.1057 &  3.9391 &  full & full       \\
 21.31 & 11.92  &  0.0075 &  0.0032 &  0.72 &  0.28  &  1.9914 &  3.7989 &   512 & countSketch\\
 22.11 & 12.89  &  0.1865 &  0.0228 &  \textbf{0.45} &  \textbf{0.11}  &  1.9893 &  3.8453 &   512 & countGauss\\
 22.42 & 12.37  &  1.6057 &  0.6437 &  1.88 &  0.67  &  2.0148 &  3.9014 &   512 &         RG\\
\hline                    
\end{tabular}
\caption{We applied CountGauss (CG), CountSketch (CW) and Random Gaussian (RG) on the TechTC300 dataset consisting of
$295$ pairs of data matrices and show the resulting mean and standard deviation for the resulting parameters such as
projection time, SVMf time (projection + SVM training time), margin (gamma) and testing error. The results are shown
over $10$-fold cross validation with $4$ repetitions and $3$ runs over the random projection matrices. Note that the mean running times
for our algorithm CG (highlighted) is faster than both CW and RG in spite of slower projection time than CW.}
\end{center}
\end{table*}
\section{Experiments}
We show experiments validating our projection operator countGauss (CG) for NMF problems on various synthetic and real-world datasets.
Also, we apply CG on the SVM problem for the TechTC300 datasets. In all of our experiments\footnote{\url{https://github.com/marinkaz/nimfa}}, we set $B=5m$.
\\
{\bf Synthetic datasets.}
%
%
Similar to ~\cite{damle2014random}, we generate the data as follows:
We set a grid of tuples $(k, m)$ such that $m/k \approx \log k$. For each tuple, we generate $500$ separable 
datasets, say X, such that they are of size $1000 \times 500$ and have nonnegative rank $k$. 
Choose matrix $U$ to have i.i.d. samples from the uniform random distribution in $[0,1]$, and be of size $d \times k$.
Also, generate matrix $V$ with the identity matrix for the top $k$ indices and the rest with i.i.d samples from the uniform
distribution. Normalize each row of the matrix $V$ to unit norm and compute the matrix product $X = U\trans{V}$.
From Figure~\ref{fig:exact_recover}, we see that the CountGauss algorithm also requires $O(k\log k)$ optimizations to find all $k$ extreme points with high probability.
We also test the algorithm in the noisy case. For that, we  generate $U$ of size $1000 \times 20$ with uniform entries in $[0,1]$ and
set the first $20$ columns of data matrix $X$ to $U$.
The remaining $190$ columns of $X$ are set to the midpoints of the $k(k+1)/2$-dimensional faces of the polytope with extreme points chosen by
the first $20$ columns of $X$. Now, we add Gaussian noise to $X$ with noise level $\sigma$, creating many spurious
extreme points. The resulting scree plot is shown in Figure~\ref{fig:exact_recover}.
\\
{\bf Flow cytometry.}
The flow cytometry (FC) data represents abundances of fluorescent molecules labeling antibodies
that bind to specific targets on the surface of blood cells. A more detailed description of
the dataset can be found in~\cite{benson2014scalable}. 
The measurements are represented as the data matrix A of size $40000 \times 5$.
Since they study pairwise interactions in the data, the Kronecker product, $X = A\otimes A$ is formed which is of size
$40000^2 \times 5^2$. 


For this dataset, we exploit the data structure as follows. For some arbitrary input vector $g$, we know that
$A \otimes A g = \trans{A} G A$ where $g=\textrm{vec}(G)$. For each random projection,
we can compute the matrix-matrix product $AG$ very efficiently and in fact do not even need to generate the matrix $G$.
For our algorithm, we do not need to
explicitly compute the Kronecker product and the complete NMF problem, including anchor selection and learning the weight coefficients, 
can be solved in a couple of seconds on an off-the-shelf desktop. As we can see from Figure~\ref{fig:flow} the results are pretty consistent from prior work~\cite{benson2014scalable}. 
The weight matrix $H$ still maintains a diagonal-like structure as previously observed.
\\
{\bf Gene expression breast cancer dataset.}
We utilize the hereditary breast cancer dataset collected by Hedenfalk et al. (2001) which consists of 
the expression levels of $3226$ genes on $22$ samples from breast cancer patients. 
The patients consist of three groups: 7 patients with a BRCA1 mutation, 8 samples with a BRCA2
mutation4, and 7 additional patients with sporadic cancers. It was analyzed using separable NMF in~\cite{damle2014random} and we similarly preprocess
the dataset by exponentiating to make the log- expression levels nonnegative and normalize the columns.
The size of the data matrix is $3222 \times 22$.
The result of applying our algorithm CG and GP are shown in Figure~\ref{fig:breast_comp}.
Notice that we get similar reconstruction error as GP while we vary the number of anchors.
\\
{\bf Climate Dataset.} 
We obtained a climate dataset which was analyzed in~\cite{tepper2015compressed}. The data size is $10512 \times 23742$.
First we present the running times and reconstruction error using SC versus QR-based algorithms and then show the corresponding results using the CG algorithm in Figure~\ref{fig:climate}. Note that CG (and SC) which is based on random projections is an order-of-magnitude faster compared to QR factorization methods. 
%
%
\\
{\bf SVM TechTC-300 Dataset.} 
We obtained the TechTC-300 dataset which is a comprehensive directory of the web. There are $295$-pairs of categories, providing a rich framework for running SVM experiments~\cite{paul2014random}. 
Each data matrix has $10,000 - 40,000$ words and $150-280$ documents. LIBSVM was used with a linear kernel and soft-margin parameter $C$ set to $500$ for all experiments and we set
the projections to $128, 256,$ and $512$.
The results are summarized in Table 1.

\section{Proof of Theorem~\ref{thm:main}} \label{sec:countsketch}

The main result of this section is

{\noindent{\bf Theorem~\ref{thm:main}} \em (Restated)
There exists an absolute constant $C>0$  such that for every $\delta\in (0, 1)$, every integer $m\geq 1$ and every matrix $U\in \R^{n\times d}$ with orthonormal columns if $B\geq \frac1{\delta^2}C d^2\cdot m$,  $S\in \mathbb{R}^{B\times n}$ is a random CountSketch matrix, and $G\in \R^{m\times B}$ and $\tilde G\in \R^{m\times n}$ are matrices of i.i.d. unit variance Gaussians, then the total variation distance between the joint distribution $G S U$ and $\tilde GU$ is less than $\delta$.
}

\begin{remark}
Note that we restrict the range of values of $m$ in Theorem~\ref{thm:main} to $[1:n^4]$. This is because if $m>n^4$, the theorem requires $B\gg \frac1{\delta}n^2$, at which  point the CountSketch matrix $S$ becomes an isometry of $\R^n$ with high probability and the theorem follows immediately. At the same time restricting $m$ to be bounded by a small polynomial of $n$ simplifies the proof of Theorem~\ref{thm:main}  notationally.
\end{remark}

Recall that a CountSketch matrix $S\in \mathbb{R}^{B\times n}$ is a matrix all of whose columns have exactly one nonzero in a random location, and the value of the nonzero element is independently chosen to be $-1$ or $+1$. All random choices are made independently. Throughout this section we denote the number of rows in the CountSketch matrix by $B$. Note that the matrix $S$ is a random variable. Let $G$ denote an $m\times B$ matrix of independent Gaussians. For an $n\times d$ matrix $U$ with orthonormal columns let  $q:\mathbb{R}^d\to \mathbb{R}_+$ denote the p.d.f. of the random variable $G_1 S U$, where $G_1$ is the first row of $G$ (all rows have the same distribution and are independent). We note that $G_1SU$ is a mixture of Gaussians. Indeed, for any fixed $S$ the distribution of $G_1SU$ is normal with covariance matrix $M:=(G_1SU)^T(G_1SU)=U^TS^TSU$. We denote the distribution of $G_1SU$ given $S$ by 
$$
q_S(x):=\frac1{\sqrt{(2\pi)^d \det M}}e^{-\frac1{2}x^TM^{-1}x}.
$$
Note that since $S$ is a random variable, $M$ is as well. With this notation in place we have for any $x\in \R^d$
\begin{equation}\label{eq:def-q}
q(x)=\expect_S\left[q_S(x)\right].
\end{equation}
Let $p:\mathbb{R}^d\to \mathbb{R}_+$ denote the pdf of the isotropic Gaussian distribution, i.e., for all $x\in \R^d$ 
\begin{equation}
p(x)=\frac1{\sqrt{(2\pi)^d}}e^{-\frac1{2}x^Tx}.
\end{equation}

We will use the following measures of distance between two distribution in the proof of Theorem~\ref{thm:main}.
\begin{definition}[Kullback-Leibler divergence]
The Kullback-Leibler (KL) divergence between two random variables $P, Q$ with probability density functions $p(x), q(x)\in \R^d$ is given by $D_{KL}(P||Q)=\int_{\R^d} p(x)\ln \frac{p(x)}{q(x)}dx$
\end{definition}

\begin{definition}[Total variation distance]
The total variation distance between two random variables $P, Q$ with probability density functions $p(x), q(x)\in \R^d$ is given by $D_{TV}(P, Q)=\frac1{2}\int_{\R^d} |p(x)-q(x)|dx$.
\end{definition}
\begin{theorem}[Pinsker's inequality]
For any two random variables $P, Q$ with probability density functions $p(x), q(x)\in \R^d$ one has $D_{TV}(P, Q)\leq \sqrt{\frac1{2}D_{KL}(P||Q)}$.
\end{theorem}

The proof of Theorem~\ref{thm:main} uses the following simple claim. 
\begin{claim}[KL divergence between multivariate Gaussians]\label{cl:kl-basic}
Let $X\sim N(0, I_d)$ and $Y\sim N(0, \Sigma)$. Then $D_{KL}(X || Y)=\frac1{2}\tr(\Sigma^{-1}-I)+\frac1{2}\ln \det \Sigma$.
\end{claim}
\begin{proof}
One has 
\begin{equation*}
\begin{split}
D_{KL}(X || Y)&=\expect_{X\sim N(0, I_d)}[-\frac1{2}X^TX+\frac1{2}X^T\Sigma^{-1}X+\frac1{2}\ln \det \Sigma]\\
&=\expect_{X\sim N(0, I_d)}[\frac1{2}X^T(\Sigma^{-1}-I)X+\frac1{2}\ln \det \Sigma]\\
&=\frac1{2}\tr(\Sigma^{-1}-I)+\frac1{2}\ln \det \Sigma,\\
\end{split}
\end{equation*}
where we used the fact that for a vector $X$ of independent Gaussians of unit variance one has $\expect_X[X^TAX]=\tr(A)$ for any symmetric $A$ (by rotational invariance of the Gaussian distribution).
\end{proof}

We let $\|A\|$ denote the operator norm of a matrix $A$, 
i.e., $\|A\| = \sup_x \frac{\|Ax\|_2}{\|x\|_2}$.  We use
\begin{claim}\label{cl:1}
For any matrix $M$ with $||I-M||<1/2$ one has $M^{-1}=(I-(I-M))^{-1}=\sum_{k\geq 0} (I-M)^k$.
\end{claim}
\begin{claim}\label{cl:2}
For any matrix $M$ with $||I-M||<1/2$ one has $\log \det M=\log\det (I-(I-M))=\sum_{k\geq 1} -\tr((I-M)^k)/k.$
\end{claim}
and 
\begin{lemma}\label{lm:frobenius-norm}
For any $U\in \R^{n\times d}$ with orthonormal columns, and $B\geq 1$, if $S$ is a random CountSketch matrix and $M=U^TS^TSU$, then $\expect_S[||M-I||_F^2] \leq O(d^2/B)$.
\end{lemma}
For the proof of Lemma~\ref{lm:frobenius-norm} see, e..g., the proof of Theorem 13 of \cite{woodruff2014sketching} with $\ell = 2$, together with the proof of Theorem 14 there applied with $\delta = 2/(\epsilon^2 B)$.

We now have:

\begin{proofof}{Theorem~\ref{thm:main}}
One has by Lemma~\ref{lm:frobenius-norm}, that for any $U\in \R^{n\times d}$ with orthonormal columns, and $B\geq 1$, if $S$ is a random CountSketch matrix and $M=U^TS^TSU$, then  $\expect_S[||M-I||_F^2]=O(d^2/B)$. By Markov's inequality $\prob_S[||I-M||^2_F>(2/\delta)\cdot O(d^2/B)]<\delta/2$. Let $\E$ denote the event that $||I-M||^2_F \leq (2/\delta)\cdot O(d^2/B)$. We condition on $\E$ in what follows. Since $B\geq \frac1{\delta^3} Cd^2 m$ for a sufficiently large absolute constant $C>1$, we have, conditioned on $\E$, that 
\begin{equation}\label{eq:fnim-bound}
||I-M||_F^2\leq (2/\delta)\cdot O(d^2/B)=(2/\delta)\cdot \delta^3/(Cm)\leq 2\delta^2/(Cm).
\end{equation}
Note that in particular we have $||I-M|| \leq ||I-M||_F < 1/2$ conditioned on $\E$  as long as $C>1$ is larger than an absolute constant. 

By Claim~\ref{cl:kl-basic} we have $D_{KL}(X || Y)=\frac1{2}\tr(I-\Sigma^{-1})+\frac1{2} \ln \det \Sigma$. We now use the Taylor expansions of matrix inverse and $\log \det$ provided by Claim~\ref{cl:1} and Claim~\ref{cl:2} to obtain

\begin{eqnarray}\label{eq:fsggrg}
D_{KL}(X || Y)&=&\frac1{2}\tr(M^{-1}-I)+\frac1{2}\ln \det M \notag \\
&=&\frac1{2}\tr\left(\sum_{k\geq 1} (I-M)^k\right)+\frac1{2}\sum_{k\geq 1} \left(-\tr((I-M)^k)/k\right) \notag \\
&=&\frac1{2}\tr\left(\sum_{k\geq 2} (I-M)^k\right)+\frac1{2}\sum_{k\geq 2} \left(-\tr((I-M)^k)/k\right)\notag \\
&=&O(\tr((I-M)^2))\text{~~~~~~~~~~~~(since $||I-M||_2\leq ||I-M||_F < 1/2$, so we have a geometric series)} \notag \\
&=&O(||I-M||_F^2) \notag \\
&=&O(\delta^2/(Cm))\text{~~~~~~~~~~~(by~\eqref{eq:fnim-bound})}\notag \\
&\leq& (\delta/4)^2/m
\end{eqnarray}
as long as $C>1$ is larger than an absolute constant. This shows that for every $S\in \E$ one has $D_{KL}(p||q_S)\leq (\delta/4)^2/m$, and thus $D_{KL}(p||\tilde q|\E])\leq (\delta/4)^2/m$, where we let $\tilde q(x):=\expect_S[q_S(x)|\E]$.

We now observe that the vectors $(G_i SU)_{i=1}^m$ and $(\tilde G_i U)_{i=1}^m$ are vectors of independent samples from distributions $q(x)$ and $p(x)$ respectively. We denote the corresponding product distributions by $q^m$ and $p^m$. Since the good event $\E$ constructed above occurs with probability at least $1-\delta/2$, it suffices to consider the distributions $\tilde q(x)$ and $p(x)$, as 
\begin{equation*}
D_{TV}(q^m, p^m)\leq \prob[\bar \E]+D_{TV}(q^m, p^m|\E)=\prob[\bar \E]+D_{TV}(\tilde q^m, p^m),
\end{equation*}
where $D_{TV}(q^m, p^m| \E)=D_{TV}(\tilde q^m, p^m)$ stands for the total variation distance between the distribution of $(\tilde G_i U)_{i=1}^m$ and the distribution of $(G_i SU)_{i=1}^m$ conditioned on $S\in \E$. 
We can now use the estimate from~\eqref{eq:fsggrg} to get
\begin{equation*}
\begin{split}
D_{TV}(\tilde q^m, p^m)&\leq \sqrt{\frac1{2}D_{KL}(p^m || \tilde q^m)}\text{~~~~(by Pinsker's inequality)}\\
&=\sqrt{\frac{m}{2}D_{KL}(p|| \tilde q)}\text{~~~~(by additivity of KL divergence over product spaces)}\\
&\leq \sqrt{\frac{m}{2}\cdot (\delta/4)^2 /m}\text{~~~~(by ~\eqref{eq:fsggrg})}\\
&\leq \delta/4.
\end{split}
\end{equation*}

\end{proofof}

We now prove (Lemma~\ref{lm:lb} below) that the result above is essentially tight.  We will need
\begin{theorem}[Example 2.2 in~\cite{wainwright_2019}, page 29]\label{thm:chi2}
Let $Y_1,\ldots, Y_m\sim N(0, 1)$ be independent Gaussian random variables, and let $Z=\sum_{i=1}^m Y_i^2$. Then for every $t\in (0, 1)$ one has 
$\prob[|Z-\expect[Z]|\geq t \cdot m]\leq 2 e^{-t^2m/8}$.
\end{theorem}

\begin{lemma}\label{lm:lb}
There exists a constant $C>1$ such that for sufficiently large $m$,  $n$ and $d$  there exists a matrix $U\in \R^{n\times d}$ with orthonormal columns such that for $B\leq  d^2m /(C\log d)$ the following conditions hold if $m=\omega(\log d)$. If $S\in \R^{B\times n}$ is a CountSketch matrix, $G\in \R^{m\times n}$ and $\tilde G\in \R^{B\times m}$ are matrices of unit variance Gaussians, then $D_{TV}(GU, \tilde G SU)=\Omega(1)$.
\end{lemma}
\begin{proof}
Define $U\in \R^{n\times d}$ as follows. For every $i=1,\ldots, d$ the $i$-th column of $U$ contains $L=\lceil \sqrt{B}/d\rceil$ nonzero coordinates, with each nonzero entry equal to $\frac1{\sqrt{L}}$, in rows with indices in $\{L\cdot (i-1)+1, L \cdot (i-1)+2, \ldots, L\cdot i\}$. We assume that $B\geq d^2$ first, and handle the case $B\leq d^2$ later. Note that the columns of the matrix $U$ defined above are indeed orthonormal, as required.

We will show that Euclidean norms
$$
\left(||G U(e_i+e_j)||_2^2\right)_{i, j=1}^d
$$ follow a distribution that is further than a constant in total variation distance from 
$$
\left(||\tilde G S U(e_i+e_j)||_2^2\right)_{i, j=1}^d. 
$$

Recall that the CountSketch matrix contains exactly one nonzero in every column: for every $k=1,\ldots, n$ the nonzero in the $k$-th column is in position $h(k)$, and the value is $\sigma(k)$.  
For a pair of indices $a, b\in \{1, 2,\ldots, Ld\}$ we write $a\sim b$ if they belong to the support of two {\bf distinct} columns. Note that for every $a\in \{1, 2,\ldots, Ld\}$ one has $a\sim b$ if and only if $\lceil a/L\rceil \neq \lceil b/L\rceil$, and therefore there are $L^2 {d \choose 2}$ pairs $a, b$ such that $a\sim b$.
Define
$$
\mathcal{E}=\{(h, \sigma) : h(a)=h(b)\text{~for exactly one pair~}a\neq b, a, b\in \{1,2,\dots, L d\}\text{~such that~}a\sim b\}.
$$
In other words, the event $\mathcal{E}$ is the event (over the choice of $h$ and $\sigma$)  that there is {\bf a unique collision} between two coordinates out of the first $L d$ that belong to supports of two distinct columns.  Since the columns of our matrix $U$ contain disjoint blocks of $L$ coordinates each, for a total of $L$ coordinates, we have that if the event $\mathcal{E}$ happens, then the CountSketch matrix does not preserve the Euclidean length of the sum of the columns in whose blocks the collision happened particularly well, which can then be detected by computing Euclidean lengths of the sums of the corresponding columns of $\tilde G S U$. We formalize this below.

One has 
\begin{equation}
\begin{split}
\prob[\mathcal{E}]&=\sum_{a \sim b} \prob[h(a)=h(b)\text{~is the only collision}]\\
&=\sum_{a \sim b} \frac1{B} \prod_{i=3}^{Ld}\left(1-\frac{i-2}{B}\right)\\
&\geq \sum_{a \sim b} \frac1{B} \left(1-\frac{Ld}{B}\right)^{Ld}\\
&\geq \frac{L^2\cdot d(d-1)}{2B}\cdot \left(1-\frac{Ld}{B}\right)^{Ld}\\
&\geq \frac1{2}e^{-1-o(1)},
\end{split}
\end{equation}
where  we used the fact that $L=\lceil \sqrt{B}/d\rceil\geq \sqrt{B}/d$, and therefore $B\leq L^2d^2$, as well as the fact that $1-1/d=1-o(1)$. In the last transition above we used the fact that $e^{-x-x^2}\leq 1-x$ for $x\in [0, 1/2]$ and that $Ld/B\leq (\sqrt{B}/d+1)d/B\leq 1/\sqrt{B}+d/B\leq 1/2$ since $B\geq d^2$ and $d$ is assumed to be larger than an absolute constant. This implies that
$$
\left(1-\frac{Ld}{B}\right)^{Ld}\geq \left(\exp\left(-\frac{Ld}{B}-\left(\frac{Ld}{B}\right)^2\right)\right)^{Ld}=\exp\left(-1-\frac{Ld}{B}\right)=e^{-1-o(1)},
$$
where in the last transition we used the fact that $\frac{Ld}{B}\leq \frac{(\sqrt{B}/d+1)d}{B}\leq \frac1{\sqrt{B}}+\frac{1}{d}=o(1)$.

  Now condition on the event $\mathcal{E}$, let $a, b$ denote the colliding pair, and let $i^*=\lceil a/L\rceil, j^*=\lceil b/L\rceil$ denote the columns that have a nonzero entry in position $a$ and $b$ respectively.  We now show that if we condition on $\mathcal{E}$, then  $||SU (e_{i^*}+e_{j^*})||_2^2$ deviates from its expectation, namely $||U(e_{i^*}+e_{j^*})||_2^2=||e_{i^*}+e_{j^*}||_2^2=2$, quite significantly. Note that 
\begin{equation*}
\begin{split}
||SU(e_{i^*}+e_{j^*})||_2^2&=\sum_{b=1}^B \left(\sum_{r\in h^{-1}(b)} \sigma_r (U_{ri^*}+U_{rj^*})\right)^2\\
&=2(L-1)\frac1{L}+\frac1{L} (\sigma_a+\sigma_b)^2\\
&=2-\frac{2}{L}+\frac1{L} (\sigma_a+\sigma_b)^2\\
&=2+\frac{2}{L} \sigma_a\sigma_b.\\
\end{split}
\end{equation*}

We thus have, 
\begin{equation}\label{eq:norm-difference}
\begin{split}
\left|||SU(e_{i^*}+e_{j^*})||_2^2-||U(e_{i^*}+e_{j^*})||_2^2\right|&=2/L\\
\end{split}
\end{equation}
We now show that this difference can be detected with high probability by observing $GU$ and $\tilde GSU$.
To that effect note that by $2$-stability of the Gaussian distribution
\begin{equation}\label{eq:sueij-norm}
||\tilde GSU(e_i+e_j)||_2^2\sim ||SU(e_i+e_j)||_2^2 \cdot \sum_{r=1}^m Y_r^2
\end{equation}
and 
\begin{equation}\label{eq:ueij-norm}
||GU(e_i+e_j)||_2^2\sim ||U(e_i+e_j)||_2^2 \cdot \sum_{r=1}^m Y_r^2,
\end{equation}
where $Y_r \sim N(0, 1)$ are independent unit variance Gaussians. This in particular implies that with high probability for all pairs $i, j\in \{1,2,\ldots, Ld\}, i\neq j,$ simultaneously the squared norms of $GU(e_i+e_j)$ and $\tilde G S U (e_i+e_j)$ are quite concentrated around their expectations. Putting this together with the bound above, we will get that the two distributions can be distinguished conditioned on $\mathcal{E}$. Indeed, by Theorem~\ref{thm:chi2} (concentration for sums of squares of independent normal random variables, i.e., $\chi^2$ random variables) we have
\begin{equation}\label{eq:concentration}
\prob\left[\left|\sum_{r=1}^m Y_r^2-m\right|\geq 5\sqrt{m\log d}\right]\leq \frac1{d^3}.
\end{equation}

Combining~\eqref{eq:concentration} with~\eqref{eq:ueij-norm}, ~\eqref{eq:sueij-norm} and~\eqref{eq:norm-difference}, we thus get 
\begin{equation}
\begin{split}
\left|||\tilde GSU(e_i+e_j)||_2^2-m\cdot ||SU(e_i+e_j)||_2^2\right|&= ||SU(e_i+e_j)||_2^2\cdot \left|\sum_{r=1}^m Y_r^2-m\right|\\
&\leq ||SU(e_i+e_j)||_2^2\cdot 5\sqrt{m\log d}\\
&\leq 20\sqrt{m\log d},
\end{split}
\end{equation}
since conditioned on $\mathcal{E}$ one has by~\eqref{eq:norm-difference} that $||SU(e_i+e_j)||_2^2\leq 2+2/L\leq 4$. At the same time we have
\begin{equation*}
\begin{split}
\left|||GU(e_i+e_j)||_2^2-m\cdot ||U(e_i+e_j)||_2^2\right|&= ||U(e_i+e_j)||_2^2\cdot \left|\sum_{r=1}^m Y_r^2-m\right|\\
&\leq 10\sqrt{m\log d}\\
\end{split}
\end{equation*}
with probability at least $1-1/d^3$ for every $i, j\in \{1, 2,\ldots, d\}$. We thus have with probability at least $1-1/d$ for all pairs of distinct indices $i, j\in \{1, 2,\ldots, d\}$ simultaneously
\begin{equation}\label{eq:238gt8wuguweg}
\left|||GU(e_i+e_j)||_2^2-2m\right| \leq 10\sqrt{m\log d}.
\end{equation}
On the other hand, conditioned on $\mathcal{E}$ there exists a pair of distinct indices $i^*, j^*\in \{1,2,\ldots, d\}$ such that 
$$
\left|||\tilde GSU(e_{i^*}+e_{j^*})||_2^2-m\cdot \left(2+\frac{2}{L}\sigma_{a} \sigma_b\right)\right|\leq 20\sqrt{m\log d}.
$$
Thus, conditioned on $\mathcal{E}$ one has, noting that $|\sigma_a\sigma_b|=1$, that
\begin{equation*}
\begin{split}
\left|||\tilde GSU(e_{i^*}+e_{j^*})||_2^2-2m\right|&\geq m\cdot \frac{2}{L}-20\sqrt{m\log d}\\
\end{split}
\end{equation*}

Since we have $L=\sqrt{B/d^2}\leq \sqrt{\frac{m}{C\log d}}$, and the rhs of the equation above can be lower bounded by
\begin{equation*}
\begin{split}
m\cdot \frac{2}{L}-20\sqrt{m\log d}&\geq 2\sqrt{Cm\log d}-20\sqrt{m\log d}\\
&\geq 40\sqrt{m\log d},
\end{split}
\end{equation*}
as long as $C>0$ is sufficiently large.  Thus, conditioned on $\E$, we have 
\begin{equation}\label{eq:238gt8}
\begin{split}
\left|||\tilde GSU(e_{i^*}+e_{j^*})||_2^2-2m\right|&\geq 40\sqrt{m\log d}
\end{split}
\end{equation}
with probability at least $1-1/d^3$.

Combining~\eqref{eq:238gt8wuguweg} with~\eqref{eq:238gt8}, we get that it is possible to distinguish between the two distributions with probability $1-o(1)$ conditioned on $\mathcal{E}$, which happens with probability $\frac1{2}e^{-1}(1-o(1))$, and hence the bound of the lemma follows under the assumption that $B\geq d^2$. If this assumption is not satisfied, then by a calculation similar to the above if $U$ is the $d\times d$ identity matrix,  the matrix $SU$ has at least one pair of columns that are identical up to a sign flip with probability at least 
\begin{equation*}
\begin{split}
1-\prod_{i=1}^d  \left(1-\frac{i-1}{B}\right)&\geq 1-\left(1-\frac{d}{B}\right)^d\\
&\geq 1-\left(1-\frac{1}{d}\right)^d \\
&\geq 1-e^{-1}.
\end{split}
\end{equation*}
Conditioned on this event, for every $m\geq 1$ the matrix $\tilde G S U$ has at least two columns that are identical up to a sign flip, something that happens with probability zero for $G  U$. 
\end{proof}

\section{A counterexample for Fast Hadamard Transforms}\label{sec:counterexample}
\begin{figure}[h!]
    \begin{center}
    \includegraphics[width=0.7\linewidth]{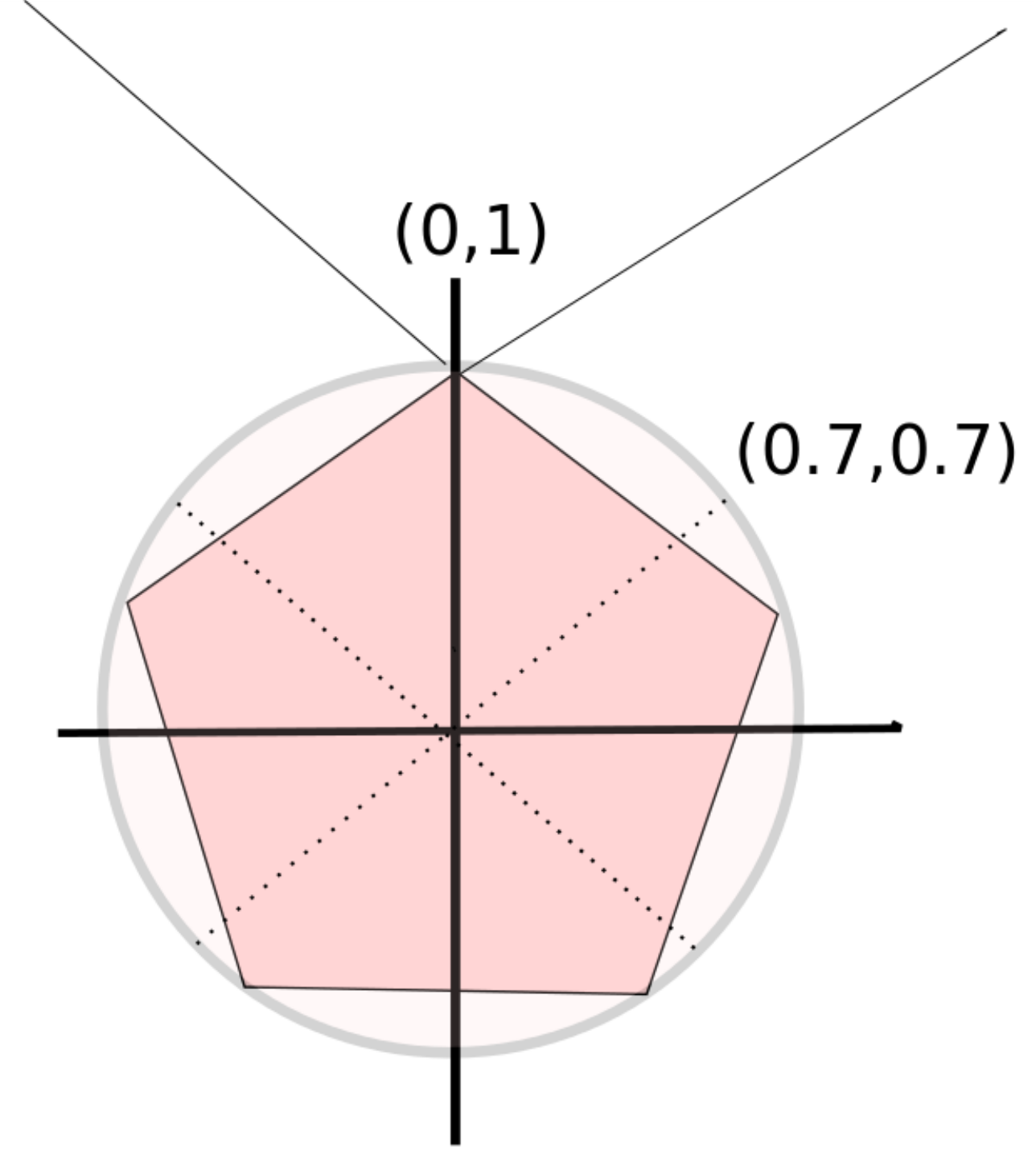}
    \end{center}
    \caption{Illustrative example with $5$ points corresponding to the vertices of a pentagon inscribed in a unit circle.
    The normal cone at $(0,1)$ is not covered by any of the vectors in $\{-1\sqrt{2},+1\sqrt{2}\}^2$.}
    \label{fig:illus_anchor}
\end{figure}    
A natural alternative
transform to try would be the $m \times d$ Subsampled Randomized Hadamard Transform
(SRHT) (see the references in Theorem 7 of \cite{woodruff2014sketching}), 
which has the form $V = P \cdot H \cdot D$, where
$P$ is a diagonal matrix with a random subset of $m$ diagonal entries
equal to $1$, and the remaining equal to $0$, $H$ is the Hadamard transform,
and $D$ is a diagonal matrix with random signs along the diagonal. Like
FastFood, the SRHT
can be applied to a $d$-dimensional vector in $O(d \log d)$ time. Note
that each row of $V$ is in the set $\{-1/\sqrt{d}, +1/\sqrt{d}\}$. 

An illustrative counterexample would be to consider a pentagon inscribed in a unit circle~\ref{fig:illus_anchor} with one
point $p$ at $(0,1)$. Each extreme point then receives $1/5$ of the circumference of the enclosing circle and so to be in 
the normal cone at $p$, one needs to have an angle in $[3\pi/10,7\pi/10]$. Hence, the second coordinate ("y") needs to have
magnitude at least $\sin(3\pi/10)$ which is larger than $1\sqrt{2}$ and so a vector in $\{-1\sqrt{2},+1\sqrt{2}\}^2$ will 
never be in it.
Generalizing this to $d$-dimensions, we could consider a convex set $C$ entirely supported on the first $2$ coordinates 
(so $0$ on the remaining coordinates). Further, we have that $C$ is a pentagon with one extreme point equal
to $(0,1,\ldots, 0)$. Now we require the second coordinate to have magnitude at least $\sin(3\pi/10)$ which is larger
than $1\sqrt{d}$ and therefore a vector in $\{-1\sqrt{d}, +1\sqrt{d}\}^d$ will never be in it.
With probability $\approx 0.16$, a random point on
the sphere will have $x_1 > 1/\sqrt{d}$ (this corresponds to one standard deviation
of an $N(0, 1/d)$ random variable), which means $\omega(N_C(p)) \approx 0.16$,
yet {\it no row of $V$} will be in $N_C(p)$, which means that even
if the condition number $\kappa$ is constant, an algorithm using the SRHT in place
of the FastFood transform will fail with probability $1$. 

\section{SVM with Random Projections }\label{sec:svm-theorem}

We require the following stronger theorem for the SVM problem~\cite{paul2014random}. \\
\begin{theorem}{Let $\epsilon \in (0, \frac{1}{2}]$ be an accuracy parameter and let $R \in \mathbb{R}^{d \times r}$ be a matrix satisfying
$\|V^T V - V^T R R^T V\|_2 \le \epsilon$
	where $V \in \mathbb{R}^{d\times \rho}$ is the orthonormal (columns) matrix of right singular vectors obtained from the SVD of $X$.
Let $\gamma^{*}$ and $\tilde{\gamma}^*$ be the 
margins obtained by solving the SVM problems using data matrices $X$ and $X R$ respectively. Then \\
\begin{center}
	$ (1-2\epsilon) \gamma^{*2} \le \tilde{\gamma}^{*2} \le ( 1 + 2\epsilon) \gamma^{*2}$
\end{center}}
\end{theorem}
\begin{proof}
We will follow a similar structure from Paul et al~\cite{paul2014random}. Define $ E  := V^T V  - V^T R R^T V$. Then for the optimal solution vectors $\alpha^*$, $\tilde{\alpha}^*$, the dual SVM objectives are given by:
\begin{align}
Z_{\textrm{opt}} &= 1^T \alpha^{*} - \frac{1}{2} \alpha^{*T} Y X X^{T} Y \alpha^{*} \\
\tilde{Z}_{\textrm{opt}} &= 1^T \tilde{\alpha}^{*} - \frac{1}{2} \tilde{\alpha}^{*T} Y X R R^{T} X^{T} Y \tilde{\alpha}^{*} 
\end{align}

	Let us first consider the objective function of the original problem at the optimal vector $\alpha^{*}$:
\begin{align}
\label{eq:zorig}
\nonumber
Z_{\textrm{opt}} &= 1^T \alpha^{*} - \frac{1}{2} \alpha^{*T} Y X X^{T} Y \alpha^{*} \\
\nonumber
                 &= 1^T \alpha^{*} - \frac{1}{2} \alpha^{*T} Y U \Sigma V^{T} R R^{T} V \Sigma U^{T} Y \alpha^{*} 
                    - \frac{1}{2} \alpha^{*T} Y U \Sigma E \Sigma U^{T} Y \alpha^{*} \\
                 &\ge \tilde{Z}_{\textrm{opt}}  - \frac{1}{2} \tilde{\alpha}^{*T} Y U \Sigma E \Sigma U^{T} Y \tilde{\alpha}^{*} 
\end{align}
where we substituted the vector $\tilde{\alpha}^*$ in the objective and utilize the fact that it results in a smaller objective value than the optimal.
Next, we consider the projected problem and lower bound it as follows:
\begin{align}
\label{eq:zproj}
\nonumber
\tilde{Z}_{\textrm{opt}} &= 1^T \tilde{\alpha}^{*} - \frac{1}{2} \tilde{\alpha}^{*T} Y X R R^{T} X^{T} Y \tilde{\alpha}^{*} \\
\nonumber
                 &= 1^T \tilde{\alpha}^{*} - \frac{1}{2} \tilde{\alpha}^{*T} Y U \Sigma V^{T} V \Sigma U^{T} Y \tilde{\alpha}^{*} 
                    - \frac{1}{2} \tilde{\alpha}^{*T} Y U \Sigma (-E) \Sigma U^{T} Y \tilde{\alpha}^{*} \\
                 &\ge Z_{\textrm{opt}}  - \frac{1}{2} \alpha^{*T} Y U \Sigma (-E) \Sigma U^{T} Y \alpha^{*} 
\end{align}
where we substituted the vector $\alpha^*$ in the objective as before and utilize the fact that it is smaller than the optimal value.
By sub-multiplicativity, we have by using the fact $V^T V = I$:
\begin{align}
\label{eq:submult}
\nonumber
\frac{1}{2} z^T Y U \Sigma G \Sigma U^T  Y z &\le \frac{1}{2} \|z^T Y U \Sigma\| \cdot \|G\|_2 \cdot \|  \Sigma U^T Y z\| \\
                                              & =  \frac{1}{2} \|G\|_2 \cdot \|z^T Y X\|_2^2 
\end{align}
for any vector $z$ and matrix $G$. 
Let us bound the following second-order term:
\begin{align*}
| z^T Y X R R^{T} X^T Y z  - z^T Y X X^T Y z| &= | z^T Y U \Sigma (V^T R R^T V  - V^T V) \Sigma U^T Y z| \\
                   &= | z^T Y U \Sigma (-E) \Sigma U^T Y z | \\
                   &\le \|E\|_2 \cdot \| z^T Y U \Sigma \|_2^2 \\
                   &= \|E\|_2 \cdot \| z^T Y X \|_2^2
\end{align*}
for any vector $z$. This gives us the following useful inequality:
\begin{align}
\label{eq:ineqobj}
\| z^T Y X \|_2^2 \le \frac{1}{1 - \|E\|_2} \|z^T Y X R\|_2^2
\end{align}

	Combining~\eqref{eq:zorig},~\eqref{eq:zproj},~\eqref{eq:submult}, and~\eqref{eq:ineqobj} for $z \in \{\alpha, \tilde{\alpha}\}$ and $G \in \{-E, E\}$, we have the 
following bounds:
\begin{align}
\nonumber
\tilde{Z}_{\textrm{opt}} &\ge Z_{\textrm{opt}} - \frac{1}{2} \|E\|_2 \cdot \|\alpha^{*T} Y X\|_2^2 \\
\nonumber
                    &= Z_{\textrm{opt}} - \|E\|_2\cdot Z_{\textrm{opt}} \\
                    &= (1 - \|E\|_2) Z_{\textrm{opt}}\\
\nonumber
Z_{\textrm{opt}} &\ge \tilde{Z}_{\textrm{opt}} - \frac{1}{2} \|E\|_2 \cdot \|\tilde{\alpha}^{*T} Y X \|_2^2 \\
\nonumber
                   &\ge  \tilde{Z}_{\textrm{opt}} - \frac{\|E\|_2}{1 - \|E\|_2}  \frac{\|\tilde{\alpha}^{*T} Y X R \|_2^2}{2} \\
\nonumber
                 &= \tilde{Z}_{\textrm{opt}} - \frac{\|E\|_2}{1 - \|E\|_2} \tilde{Z}_{\textrm{opt}} \\
                 &= ( 1  - \frac{\|E\|_2}{1 - \|E\|_2} ) \tilde{Z}_{\textrm{opt}}   
\end{align}
%
The bounds follow by using the following relations, 
$Z_{\textrm{opt}} = \frac{1}{2 \gamma^{*2}}$ and $\tilde{Z}_{\textrm{opt}} = \frac{1}{2 \tilde{\gamma}^{*2}}$:
\begin{align}
(1 - \frac{\|E\|}{1 - \|E\|_2}) \gamma^{*2} \le \tilde{\gamma}^{*2} \le \frac{1}{1 - \|E\|_2} \gamma^{*2}
\end{align}
Notice that we cannot now trivially project the data to zero ($XR$ to the zero matrix) which would have been acceptable if we had used the weaker version of the theorem as stated in~\cite{paul2014random}. 
\end{proof}

\section{Discussion}
 We have presented an efficient way to multiply by a Gaussian matrix, without actually computing the dense
 matrix product. Theorem~\ref{thm:main} provides our theoretical guarantees on this much faster transform,
 showing it has low variation distance to multiplication by a dense Gaussian matrix.

 Our transform is useful in a surprising number of applications --- here we apply our transform to NMF and SVM.
 The classical way of speeding up Gaussian transforms via the Fast Hadamard or FFT does not work in our setting since it misses large sections of the sphere.

 Our experiments on synthetic and real-world datasets for NMF showed that 
 the results obtained by our algorithm were on par with the state-of-the-art NMF algorithms such as
 SC, XRAY and SPA. In particular, for synthetic problems, we showed similar anchor recovery performance as random projection (GP) of 
~\cite{damle2014random} both in the noiseless and noisy cases. Also, the performance was remarkably similar to  GP when applied on the breast cancer dataset and also picked up activation patterns which might be of biological interest as previously noted in flow cytometry problems.
 Experiments on document classification tasks using the popular SVM formulation revealed that the new projection leads to faster 
 SVM solutions than previous methods. Previously, it was shown that while CountSketch led to faster projection times it did not lead to overall
 faster training time and in fact was found to be slower than random Gaussian projections (RG). Our new countGauss projection fixes this by sacrificing projection time compared to countSketch projection but leads to an overall faster SVM training time and thereby beats both random Gaussian and CountSketch-based SVM algorithms~\cite{paul2014random}.   
 We note that in practice for SVM, solution accuracy may be of critical importance rather than computation time and in these
 scenarios random projection based algorithms can be used to explore the optimal settings of the SVM parameters such as soft-margin. 
 In our experiments (not shown) we noticed that these lead to faster training times while not sacrificing test accuracy.

\section*{Acknowledgements}
Vamsi P. would like to acknowledge support from the following grant: NSF-IIP-1346452 and also RDI$^2$ at Rutgers University where this work was initiated. David W. would like to thank Sitan Chen and Jerry Li for helpful discussions, and would like to acknowledge the support from XDATA program of the Defense Advanced Research Projects Agency (DARPA), administered through Air Force Research Laboratory contract FA8750-12-C-0323.


\end{document}